\documentclass[10pt,a4paper]{article}

\usepackage{graphicx,amsmath,amssymb,fixmath,amsthm}

\newcommand{\Ord}[1]{ {O}( #1 )}
\newcommand{\R}{{\Bbb R}}

\newcommand{\N}{{\Bbb N}}

\newcommand{\D}[2]{ \ensuremath{ \frac{d #1 }{d #2 } }}

\newcommand{\vect}[1]{\ensuremath{ \mathbold #1 } }

\newtheorem{theorem}{Theorem}[section]
\newtheorem{lemma}[theorem]{Lemma}

\theoremstyle{definition}
\newtheorem{definition}[theorem]{Definition}

\theoremstyle{remark}

\newtheorem{proposition}[theorem]{Proposition}

\date{}

\begin{document}

\title{Tropicalization and tropical equilibration of chemical reactions}

\author{V.~Noel$^1$, D.~Grigoriev$^2$, S.~Vakulenko$^3$ and O.~Radulescu$^4$ \\ \\
 \small  $^1$ IRMAR UMR 6625, University of Rennes 1, Rennes, France, \\
 \small  $^2$ CNRS, Math\'ematiques, Universit\'e de Lille, 59655, Villeneuve d'Ascq, France, \\
\small  $^3$ Institute of Problems of Mechanical Engineering, St.Petersburg, Russia, \\
 \small  $^4$ DIMNP UMR CNRS 5235, University of Montpellier 2, Montpellier, France.}



\maketitle

\centerline{\bf Abstract}

Systems biology uses large networks of biochemical reactions to model the
functioning of biological cells from the molecular to the cellular scale.
The dynamics of dissipative reaction networks with many well separated time
scales can be described as a sequence of successive equilibrations of
different subsets of variables of the system. Polynomial systems with
separation are equilibrated when at least two monomials, of opposite signs,
have the same order of magnitude and dominate the others. These equilibrations
and the corresponding truncated dynamics, obtained by eliminating the dominated
terms, find a natural formulation in tropical analysis and can be used for
model reduction.

{\bf Keywords:}  Tropical analysis, asymptotic analysis, chemical kinetics, systems biology.

{\bf AMS subjects:} Primary 14T05,  92C40, 92C42; Secondary 14M25.

\section{Introduction.}

Systems biology develops biochemical dynamic models of various cellular processes
such as signalling, metabolism, gene regulation. These models can reproduce complex
spatial and temporal dynamic behavior observed in molecular biology experiments.
The dynamics of multiscale, dissipative, large biochemical models,
can be reduced to that of simpler models, that were called dominant
subsystems \cite{radulescu2008robust,gorban2009asymptotology,gorban-dynamic}.
Simplified, dominant subsystems contain less parameters and are more easy to analyze.

The notion of dominance is asymptotic and a natural mathematical framework to capture
multiple asymptotic relations is the tropical analysis.
Motivated by applications
in mathematical physics \cite{litvinov1996idempotent},
systems of polynomial equations \cite{sturmfels2002solving},
etc.,  tropical analysis uses a change of
scale to transform nonlinear systems into piecewise
linear systems.

In this paper we provide some mathematical justifications for possible applications of
the idea of tropicalization to systems biology models.

\section{Tropicalized chemical kinetics}
In chemical kinetics, the reagent concentrations satisfy
ordinary differential equations:
\begin{equation}
\D{x_i}{t} = F_i (\vect{x}), \, 1 \leq i \leq n.
\label{mainsystem}
\end{equation}
Rather generally, the rates are rational functions of the concentrations and read
\begin{equation}
F_i (\vect{x})= P_i (\vect{x})/Q_i(\vect{x}),
\label{rational}
\end{equation}
where
$P_i(\vect{x}) = \sum_{\vect{\alpha} \in A_i} a_{i,\vect{\alpha}} \vect{x}^\vect{\alpha}$,
$Q_i(\vect{x}) = \sum_{\vect{\beta} \in B_i} b_{i,\beta} \vect{x}^\vect{\beta}$,
are multivariate polynomials. Here
$\vect{x}^\vect{\alpha} = x_1^{{\alpha}_1} x_2^{{\alpha}_2} \ldots x_n^{{\alpha}_n}$,
$\vect{x}^\vect{\beta} = x_1^{\beta_1} x_2^{\beta_2} \ldots x_n^{\beta_n}$, $a_{i,\vect{\alpha}}, b_{i,\vect{\beta}}$, are
nonzero real numbers, and $A_i, B_i$ are finite subsets of $\N^n$.

The special case of mass action kinetics is represented by
\begin{equation}
F_i (\vect{x})= P_i^+(\vect{x}) - P_i^-(\vect{x}),
\label{laurent}
\end{equation}
where $P_i^+(\vect{x})$, $P_i^-(\vect{x})$ are positive coefficients polynomials,
$P_i^{\pm}(\vect{x}) =$
$\sum_{\vect{\alpha} \in A_i^\pm} a_{i,\vect{\alpha}}^\pm \vect{x}^\vect{\alpha},$
$a_{i,\vect{\alpha}}^\pm > 0$, and $A_i^\pm$ are finite subsets of $\N^n$.

In multiscale biochemical systems, the various monomials defining reaction rates have
different orders, and at a given time, there is only one or a few dominating terms.
Therefore, it could make sense to replace polynomials with positive real
coefficients  $\sum_{\vect{\alpha} \in A} a_{\vect{\alpha}} \vect{x}^\vect{\alpha}$, by max-plus polynomials
$\exp ( \max_{\vect{\alpha} \in A} (log( a_{\vect{\alpha}}) + < log(\vect{x}), \vect{\alpha} > ))$.

This heuristic can be used to associate a piecewise-smooth model
to the system of rational ODEs \eqref{mainsystem}, in two different ways.

The first method was proposed in \cite{SASB2011} and can be applied to any rational
ODE system defined by \eqref{mainsystem},\eqref{rational}:
\begin{definition}
We call complete tropicalization  of the smooth ODE
system \eqref{mainsystem},\eqref{rational} the following piecewise-smooth system:
\begin{equation}
\D{x_i}{t} = Dom P_i (\vect{x}) / Dom Q_i(\vect{x}),
\label{tcomplete}
\end{equation}
\noindent where $Dom \left( \sum_{\vect{\alpha} \in A_i} a_{i,\vect{\alpha}}  \vect{x}^\vect{\alpha} \right) =
sign(a_{i,\vect{\alpha}_{max}}) \exp (\max_{\vect{\alpha} \in A_i} ( log( |a_{i,\vect{\alpha}}|) + < \vect{u} , \vect{\alpha} > ))$.
Here $\vect{u} = (log x_1,\ldots,log x_n)$, $< \vect{u} , \vect{\alpha} >$ denotes the dot product, and
$a_{i,\vect{\alpha}_{max}}$, $\vect{\alpha}_{max}\in A_i$ denotes the coefficient
of the monomial for which the maximum is attained. In simple words, $Dom$ renders the monomial
of largest absolute value, with its sign.
\end{definition}
The second method,proposed in \cite{savageau2009phenotypes}, applies to the systems \eqref{mainsystem},\eqref{laurent}.
\begin{definition}
We call two terms tropicalization  of the smooth ODE
system \eqref{mainsystem},\eqref{laurent} the following piecewise-smooth system:
\begin{equation}
\D{x_i}{t} = Dom P_i^+ (\vect{x}) - Dom P_i^- (\vect{x}),
\label{2terms}
\end{equation}
\end{definition}
The two-terms tropicalization was used in \cite{savageau2009phenotypes} to analyse the
dependence of steady states on the model parameters. The complete tropicalization
was used for the study of the model dynamics and for the model reduction \cite{SASB2011,radulescu2012red}.

For both tropicalization methods, for each occurrence of the Dom operator, one can introduce
a tropical manifold, defined as the subset of $\R^n$ where the maximum in $Dom$ is attained by at least two
terms. For instance, for $n=2$, such tropical manifold is made of points, segments connecting these points, and half-lines. The tropical manifolds in such an arrangement decompose the space into sectors, inside which one monomial
dominates all the others in the definition of the reagent rates.
The study of this arrangement give hints on the possible steady states and attractors, as well as on their
bifurcations \cite{SASB2011}.




\section{Tropical equilibration and permanence}
In the general case, the tropicalization heuristic is difficult to justify by rigorous estimates.
However, this is possible when the polynomials defining the rhs of the ODE system have dominant monomials,
much larger than the other monomials.

To simplify, let us consider the class of polynomial systems, corresponding
to mass action law chemical kinetics :
\begin{equation}
\frac{dx_i}{dt} = \sum_{j=1}^{M_i} P_{ij} \vect{x}^{\vect{\alpha}_{ij}}
\label{chem}
\end{equation}
\noindent
where $\vect{\alpha}_{ij}$ are multi-indices, $P_{ij}$ are rate constants.

In order to introduce orders, we consider that coefficients $P_{ij}$ are integer powers of a small positive parameter $\epsilon$:
\begin{equation}
P_{ij}(\epsilon)=\epsilon^{\gamma_{ij}} \bar P_{ij}.
\label{eps}
\end{equation}
where $\bar P_{ij}$ do not depend  or are $\Ord{1}$ on $\epsilon$.

We also suppose that the cone ${\bf R}_{> }=\{x: x_i  \ge 0 \}$ is invariant under dynamics (\ref{chem})
and initial data are positive:
$$
x_i(0) > \delta > 0.
$$
The terms (\ref{eps}) can have different signs, the ones with $\bar P_{ij} >0$ are production terms,
and those with $\bar P_{ij} <0$ are degradation terms.

From the biochemical point of view, the choice (\ref{eps}) is justified by the
fact that biochemical processes have many, well separated concentration and time
scales. The orders of different monomials defining the system \eqref{chem} are set
by orders of the parameters but also  by the orders of the concentrations
variables $x_i$. We therefore use a renormalization :
\begin{equation}
  x_i=\epsilon^{a_i} \bar x_i.
\label{renorm}
\end{equation}
where $a_i$ are unknown powers chosen such that $\bar x_i$ are bounded uniformly in $\epsilon$
(we will see later when this choice is possible).


We seek for renormalization exponents $a_i$ such that only a few terms dominate all the others,
for each $i$-th equation \eqref{chem} as $\epsilon \to 0$.  Let us denote the number of terms with minimum
degree in $\epsilon$ for $i$-th equation as $m_i$. Naturally, $1 \le m_i \le M_i$.
After renormalization, we remove all small terms that have smaller orders in $\epsilon$  as $\epsilon \to 0$.
We can call this procedure {\em tropical removing}. The system obtained can be named {\em tropically truncated system}.

Let us denote $\alpha_l^{ij}$ the $l^{th}$ coefficient of the multi-index $\vect{\alpha}_{ij}$.
If all $m_i=1$ then we have the following truncated system
\begin{equation}
\frac{d\bar x_i}{dt} =\epsilon^{\mu_i} F_{i}(\bar {\bf x}), \quad F_{i}(\bar {\bf x})=P_{ij(i)} \bar {\bf x}^{\vect{\alpha}^{ij(i)}},
\label{chemshort}
\end{equation}
where $j(i)$ is the index of the unique term with minimum degree in $\epsilon$,
\begin{equation}
\mu_i=\gamma_{i j(i)} + \sum_{l=1}^n \alpha_l^{ij(i)} a_l - a_i,
\label{mu1a}
\end{equation}
and
\begin{equation}
\mu_i < \gamma_{i j} + \sum_{l=1}^n \alpha_l^{ij} a_l - a_i \quad for \ all \  j \ne j(i).
\label{mu2}
\end{equation}
 If all $m_i=2$, in order to find possible renormalization exponents $a_i$,  it is necessary to resolve a family of linear programming problem. Each problem is defined by a set of pairs $(j(i), k(i))$ such that $j(i) \ne k(i)$. We define $\mu_i$ by
\begin{equation}
\mu_i=\gamma_{i j(i)} + \sum_{l=1}^n \alpha_l^{ij(i)} a_l - a_i=\gamma_{i k(i)} + \sum_{l=1}^n \alpha_l^{ik(i)} a_l -a_i
\label{mu11}
\end{equation}
and obtain the system of the following inequalities
\begin{equation}
\mu_i \le \gamma_{ij} + \sum_{l=1}^n \alpha_l^{ij} a_l -a_i \quad for \ all  \ j \ne j(i), k(i).
\label{mu21}
\end{equation}

In order to define more precisely the separation between various terms,
we use the permanency concept, borrowed from ecology
(the Lotka-Volterra model, see for instance \cite{y1996global}).
\begin{definition}
The system \eqref{chem} is permanent, if there are two constants $C_{-} >0$ and $C_{+} > 0$,
a set of renormalization exponents $a_i$, and
 a function $T_0$, such that after the renormalized variables \eqref{renorm} satisfy
\begin{equation}
 C_{-} < \bar x_i(t) < C_{+},   \quad for \ all \ t > T_0 ( x(0)) \ and \ for \ every \ i.
\label{perm}
\end{equation}
We assume that $C_{\pm}$ and $T_0$ are uniform in (do not depend on) $\epsilon$ as $\epsilon \to 0$.
\end{definition}

For permanent systems, we can obtain some results justifying the two
procedures of tropicalization.
\begin{proposition}
\label{comparison}
Assume that system \eqref{chem} is permanent.
Let $x$, $\hat x$ be the solutions to the
Cauchy problem for (\ref{chem}) and
(\ref{tcomplete}) (or (\ref{2terms})), respectively, with the same initial
data:
$$
  x(0)=\hat x(0).
$$
Then the difference $y(t)=x(t) - \hat x(t)$ satisfies  the estimate
\begin{equation}
     | y(t) |  < C_1 \epsilon^{\gamma} \exp(bt), \quad \gamma > 0,
\label{estdif}
\end{equation}
 where the positive constants $C_1, b$ are uniform in $\epsilon$.
If the original system (\ref{chem}) is structurally stable in the domain
$\Omega_{C_{-}, C_{+}}=\{x: C_{-} < |x| < C_{+} \}$,
then the corresponding tropical systems (\ref{tcomplete}) and (\ref{2terms})
are also permanent and there is an orbital topological equivalence
$\bar x = h_{\epsilon}(x)$ between the trajectories $x(t)$ and $\bar x(t)$
of the corresponding Cauchy problems. The homeomorphism $h_{\epsilon}$ is close to the identity as $\epsilon \to 0$.
\end{proposition}
The proof of the estimate (\ref{estdif}) follows immediately by the Gronwall lemma.
The second assertion follows directly
from the definition of  structural stability which means that orbits of the dynamical system
are smoothly deformed under small perturbations.

Permanency property is not easy to check.
The following straightforward lemma gives a necessary condition
of permanency of the system \eqref{chem}.
\begin{lemma}
{Assume a tropically truncated system is permanent.
Then, for each $i \in \{1,\ldots, n \}$, the $i$-th equation
of this system contains at least two terms. The terms should
have different signs for coefficients $p_{ij}$, i.e., one term
should be a production one, while  another term should be
a degradation term.}
\label{lemmaequil}
\end{lemma}

\begin{proof}
Let us suppose that $m_i=1$ for some $i$, or $m_i >1$, but all terms have the same sign $s$.
Let us consider this equation.  Then one has, for $s=1$,
$$
\frac{dx_i}{dt} >  \epsilon^{\mu_i} \delta_i(C^{-}, C^{+}) > 0.
$$
Therefore, $x_i(t) > \delta t + x_i(0)$ and the system cannot be permanent.
If $s=-1$, then
$$
\frac{dx_i}{dt} <  -\epsilon^{\tilde \mu_i}\delta_i(C^{-}, C^{+}) <0.
$$
Again it is clear that the system cannot be permanent.
\end{proof}

We call ``tropical equilibration'', the condition in Lemma \ref{lemmaequil}. This condition means
that permanency is acquired only if at least two terms of different signs have the maximal order,
for each equation of the system \eqref{chem}.

The tropical equilibration condition can be used to determine the renormalization exponents, by the
following algorithm.

{\em Step 1}. For each $i$ let us choose a pair $(j(i), k(i))$ such that $j, k \in \{1,\ldots, M_i\}$ and
$j < k$. The sign of the corresponding terms should be different.

{\em Step 2}. We resolve the linear system of algebraic equations
\begin{equation}
\gamma_{i j(i)} - \gamma_{i k(i)}= -\sum_{l=1}^n \alpha_l^{ij(i)} a_l  + \sum_{l=1}^n \alpha_l^{ik(i)} a_l,
\label{alg}
\end{equation}
for $a_l$, together with the inequalities (\ref{mu21}).

Notice that although that Step 2 has polynomial complexity, the tropical equilibration problem
has a number of choices that is exponential in the number of variables
at Step 1.

Assume that, as a result of this procedure, we obtain the two terms {\em toric system}
\begin{equation}
\frac{d\bar x_i}{dt} =\epsilon^{\mu_i} (F_{i}^+(\bar {\bf x}) - F_{i}^-(\bar {\bf x})), \quad F_{i}^{\pm}=P_{ij^{\pm}} \bar {\bf x}^{\vect{\alpha}_{\pm}^{ij}}.
\label{chemtrop1}
\end{equation}
One can expect that, in a "generic" case\footnote{supposing
that multi-indices $\vect{\alpha}_{ij}$ are chosen uniformly, by generic we understand
almost always except for cases of vanishing probability,
see also \cite{gorban-dynamic}}, all $\mu_i$ are mutually different, namely
\begin{equation}
\mu_1 < \mu_2 < ... < \mu_{n-1} <\mu_n.
\label{chain}
\end{equation}

We can now state a sufficient condition for permanency.
Let us consider the first equation (\ref{chemtrop1}) with $i=1$ and let us denote $y=\bar x_1, z=(\bar x_2, ..., \bar x_n)^{tr}$.
In this notation, the first equation becomes
\begin{equation}
\frac{dy}{dt} =f(y)=b_1(z) y^{\beta_1} - b_2(z) y^{\beta_2}, \quad b_1, b_2 > 0, \quad \beta_i \in {\bf R}.
\label{chemtrop}
\end{equation}
According to \eqref{chain} here $z(t)$  is a slow function of time and thus we can suppose that $b_i$ are constants (this step will be rendered rigorous at the end of this section, by using the concept of invariant
manifold and methods from \cite{henry1981geometric}).
The permanency property can be then checked in an elementary way.
\begin{lemma}{
Equation \ref{chemtrop} has the permanence property if and only if
 $$
\beta_1 < \beta_2.
$$
For fixed $z$ in these cases we have
$$
y(t,z) \to y_0(z) \quad as \ t \to \infty.
$$
}
\label{permcrit}
\end{lemma}

\begin{proof}
Consider the function  $f(y) = b_1 y^{\beta_1} - b_2 y^{\beta_2}$.
Under the condition $\beta_1 < \beta_2$, $f$ is negative for sufficiently large $y>0$,
and positive for sufficiently small $y>0$.
Moreover, $f$ has a single positive root $y_1$ on $(0, +\infty)$.
Therefore, all the trajectories of $dy/dt=f(y)$ tends to $y_1$ as $t \to \infty$ and,
for any $\delta >0$, the interval $(y_1-\delta, y_1+\delta)$ is a trapping domain.
This proves the permanency.
\end{proof}

Let us note that tropical equilibrations with permanency imply the
existence of invariant manifolds. This allows to reduce the number of variables of the model while preserving good accuracy in the description of the dynamics. The following Lemma
is useful in this aspect.
\begin{lemma}{
Consider the system
\begin{equation}
\frac{dy}{dt} =f(y,z)=b_1(z) y^{\beta_1} - b_2(z) y^{\beta_2}, \quad b_1, b_2 > \delta_1 >0, \quad \beta_i \in {\bf R}.
\label{chemtrop4}
\end{equation}
\begin{equation}
\frac{dz}{dt} =\lambda F(y, z),
\label{Zdt}
\end{equation}
where $z \in {\bf R}^m$, $\lambda >0$ is a parameter
 and the function $F$ enjoys the following properties. This function lies in an H\" older class
$$
F \in C^{1+r}, \quad r > 0,
$$
and the corresponding norms are uniformly bounded  in $\Omega=(0, +\infty) \times W$, for some open domain $W \subset {\bf R}^m$:
$$
 |F|_{C^{1+r}(\Omega)} < C_2.
$$
Assume that the condition  of Lemma \ref{permcrit} holds.
We also suppose that $b_i$ are smooth functions of $z$ for all $z$ such that $|z| >\delta_0 >0$. Assume that
$z \in W$ implies $|z| > \delta_0$.

Let $y_1(z)$ be the unique solution of $f(y,z)=0$.

Then, for sufficiently small $\lambda < \lambda_0(C_2, b_1, b_2, \beta_1, \beta_2, \delta_0, \delta_1)$ equations (\ref{chemtrop4}), (\ref{Zeq})
have a locally invariant and locally attracting manifold
\begin{equation}
y =Y(z, \lambda),  \quad Y \in C^{1+r}(W),
\label{Zeq}
\end{equation}
and $Y$ has the asymptotics
\begin{equation}
Y(z, \lambda)=y_1(z) + \tilde Y,  \quad \tilde Y \in C^{1+r}(W),
\label{Zeq1}
\end{equation}
where
\begin{equation}
|\tilde Y(z, \lambda)|_{C^{1+r}(W)} < C_s\lambda^{s},  \quad s >0.
\label{Zeq2}
\end{equation}
}
\label{L3.6}
\end{lemma}
\begin{proof}
This lemma can be derived from
Theorem 9.1.1 from (\cite{henry1981geometric}, Ch. 9).
\end{proof}

The generic situation described by the conditions \eqref{chain} leads to trivial ``chain-like''
relaxation towards a point attractor, provided that we have permanency at each step.
More precisely, all the variables have separate timescales and dissipative dynamics.
The fastest variable relaxes first, then the second fastest one, and so forth, the chain
of relaxations leading to a steady state.

The following theorem describes a less trivial situation, when some timescales are not totally separated and limit cycles are possible.
\begin{theorem}
Assume $\mu_1 < \mu_2 < ... < \mu_{n-1} \leq \mu_n$ holds.

i) If the procedure, described above, leads to the permanency property at each step, where $i=1,2,..., n-2$,
and if the successive application of the lemma \ref{L3.6} for the tropically truncated toric system  (\ref{chemtrop1}) uniquely defines the locally invariant  smooth manifold
\begin{equation} \label{invman}
\bar x_i=X_i(\bar x_{n-1}, \bar x_{n}),  \quad X_i \in C^{1+r}, \quad i=1,...,n-2,\quad r >0,
\end{equation}
as the unique stable hyperbolic equilibrium of the tropically truncated system \eqref{chemtrop1}.
Then,  the original system has an invariant manifold close to \eqref{invman}
\begin{equation} \label{invman2}
\bar x_i=X_i(\bar x_{n-1}, \bar x_{n})+ \phi(\bar x_{n-1}, \bar x_{n},\epsilon),  \quad i=1,...,n-2.
\end{equation}
where the corrections $\phi_i$  satisfy
$$
 |\phi_i(\cdot,\cdot,\epsilon)|_{C^{1+r}} \to 0 \quad  (\epsilon \to 0).
$$

ii) If the procedure, described above, leads to the permanency property at each step, where $i=1,2,..., n-2$,
and the last two equations of the tropically truncated system
have a globally attracting hyperbolic rest point or globally attracting hyperbolic limit cycle,
then the tropically truncated system is permanent and has an attractor of the same type.
 Moreover, for sufficiently small $\epsilon$ the initial system also is permanent for initial data from  some
appropriate domain $W_{\epsilon, a, A}$ and has an analogous
 attracting hyperbolic rest point (limit cycle) close to the attractor of the truncated system.

iii) If the rest point (cycle) is not globally attracting,
then we can say nothing on permanency but, for sufficiently small $\epsilon$,  the initial system  still has an analogous
attracting hyperbolic rest point (limit cycle) close to the attractor of truncated system and the same topological structure.
\label{bigtheorem}
\end{theorem}
\begin{proof}

{\bf i)}   This follows from  Lemma \ref{L3.6}, which can be applied inductively, step by step.

{\bf ii)}  Suppose that the tropically truncated system (TTS) has a globally attracting compact invariant set $\mathcal A$. Let $\Pi$ be an open neighborhood
of this set. We can choose this neighborhood as a box that contains $\mathcal A$. Then, for all initial data
$x(0)$, the corresponding trajectory $x(t), x(0)$ lies in $\Pi$ for all $t > T_0(x_0, \Pi)$.  Therefore, our TTS is permanent.
Here we do not use the fact that the cycle (rest point) is hyperbolic.

 Permanency of the initial system follows from hyperbolicity of $\mathcal A$. Hyperbolic sets are persistent (structurally stable \cite{ruelle1989chaotic}). Since this set is globally attracting, all TTS is structurally stable (as a dynamical system).
This implies that the initial system has a  hyperbolic attractor close to $\mathcal A$, since initial system is a small perturbation
 of TTC in $\Pi$.

{\bf iii)} If the set $\mathcal A$ is only locally attracting, the last assertion of the Theorem follows from persistency of hyperbolic sets.
\end{proof}

{\em Remark.}
Theorem \ref{bigtheorem} implicitly supposes that all fast variables $x_i$, $i=1,2,..., n-2$ can
be expressed as functions of the remaining slow variables $x_{n-1},x_n$. It does not consider the
situation when the successive application of the lemma \ref{L3.6} leads to degenerate equilibria. This
situation typically occurs when the tropically truncated system has conservation laws, i.e.
linear combinations of the fast variables are invariant with respect to the truncated fast dynamics.
This case, asking for variable aggregation and new slow variables will be discussed in detail elsewhere.

\section{Geometry of tropical equilibrations}

In this section we provide a geometrical interpretation of tropical equilibrations. We consider
networks of biochemical reactions with mass action kinetic laws.
Each reaction between reagents $A_i$
is defined as $$\sum_i \vect{\alpha}_{ji} A_i \rightarrow \sum_{k} \beta_{jk} A_k.$$
The stoichiometric vectors  $\vect{\vect{\alpha}_{j}} \in \N^n$, $\vect{\beta_{j}} \in \N^n$ have coordinates
$\vect{\alpha}_{ji}$ and $\beta_{jk}$ and define which species are consumed and produced by the reaction $j$
and in which quantities.
The mass action law means that reaction rates are monomial functions and read
 \begin{equation}
 R_j(\vect{x}) = k_j \vect{x^{\vect{\alpha}_{j}}}.
 \label{rrateone}
 \end{equation}
 where $k_j >0$ are kinetic constants.
The network dynamics is described as follows
 \begin{equation}
 \D{\vect{x}}{t} = \sum_j k_j ({\beta_{ji}} - {\alpha_{ji}})  \vect{x}^{\vect{\alpha_{j}}}.
 \label{massaction}
 \end{equation}
After parameters and variables rescaling, $k_j = \bar k_j \epsilon^{\gamma_j}$,
$\vect{x} = \bar{\vect{x}} \epsilon^{\vect{a}}$ we obtain
 \begin{equation}
 \D{\bar{x}_i}{t} = (\sum_j \epsilon^{\mu_j} k_j ({\beta_{ji}} - {\alpha_{ji}})  {\bar{\vect{x}}}^{\vect{\alpha_{j}}})\epsilon^{-a_i},
 \label{massactionrescaled}
 \end{equation}
where 
\begin{equation}
\mu_j = \gamma_j + <\vect{a},\vect{\alpha_j}>.
\label{muj}
\end{equation}

\begin{definition} \label{eqreactions}
Two reactions $j$, $j'$ are equilibrated on the species $i$ iff:

i) $\mu_j = \mu_{j'}$,

ii) $(\vect{\beta_{j}} - \vect{\alpha_{j}})_i (\vect{\beta_{j'}} - \vect{\alpha_{j'}})_i < 0$,

iii) $\mu_k \geq \mu_j$ for any reaction $k \neq j,j'$, such that
$(\vect{\beta_{k}} - \vect{\alpha_{k}})_i \neq 0$.
\end{definition}

{\em Remarks.}
Definition \ref{eqreactions} ensures the conditions of Lemma \ref{lemmaequil} and is thus equivalent to
tropical equilibration of the species $i$.

According to \eqref{muj} and Definition \ref{eqreactions}, the equilibrations correspond to vectors $\vect{a} \in R^n$ where the minimum in the definition of the piecewise-affine function
$f_i(\vect{a}) = \min_j (\gamma_j + <\vect{a},\vect{\alpha_j}>)$ is attained at least twice.

Let us consider the equality $\mu_j = \mu_{j'}$. This represents the equation of a
$n-1$ dimensional hyperplane of $\R^n$, orthogonal to the vector $\vect{\alpha_j} - \vect{\alpha_{j'}}$:

\begin{equation}
\gamma_j + <\vect{a},\vect{\alpha_j}> = \gamma_{j'} + <\vect{a},\vect{\alpha_{j'}}>
\label{lines}
\end{equation}

For each species $i$, we consider the set of reactions ${\mathcal R}_i$
that act on this species, namely the reaction $k$ is in ${\mathcal R}_i$ iff
$(\vect{\beta_{k}} - \vect{\alpha_{k}})_i \neq 0$.
The finite set ${\mathcal R}_i$ can be characterized by the corresponding set
of stoichiometric vectors $\vect{\alpha_k}$. 

The set of points of $\R^n$ where at least two reactions equilibrate on the species
$i$ corresponds to the places where the function $f_i$ is not locally affine (the minimum
in the definition of $f_i$ is attained at least twice).

For each species, we also define the Newton polytope ${\mathcal N}_i$, that is the
convex hull of the vectors $\vect{\alpha_k}, k\in {\mathcal R}_i$.
The hyperplanes defined by \eqref{lines} and corresponding to equilibrations of two reactions
on the same species $i$ are orthogonal to edges of the Newton polytope ${\mathcal N}_i$.
${\mathcal N}_i$ is also the Newton polytope of the polynomial
$P_i(\vect{x}) = \sum_j k_j (\beta_{ji} - \alpha_{ji})  \vect{x}^{\vect{\alpha_{j}}}$ that defines the rhs of the ordinary differential
equation satisfied by the species $i$.

We can now state the following

\begin{proposition}
{\em There is a bijection between the locus ${\mathcal T}_i$ of vectors $\vect{a}$ where the min-plus polynomial
$f_i(\vect{a})$ is not linear and the tropical manifold of the polynomial $P_i(\vect{x})$ that
defines the rhs of the ordinary differential equation satisfied by the species $i$. The reaction
equilibrations correspond to vectors $\vect{a}$ included in ${\mathcal T}_i$ but satisfying also the
condition ii) of Definition \ref{eqreactions}.               }
\end{proposition}

{\em Remarks.} This property can be used to put into correspondence reaction equilibrations and
slow invariant manifolds. Indeed, if a reaction equilibration exists, this leads to a
slow manifold that  is close to some parts of the tropical manifold
of $P_i(\vect{x})$. For instance, a reaction equilibration described by \eqref{lines}
will correspond to an invariant manifold close to a hyperplane orthogonal to
$\vect{\alpha_j} - \vect{\alpha_{j'}}$. The condition ii) of Definition \ref{eqreactions}
is needed for equilibrium (the equilibrated reactions have
to have opposite effects on the species $i$, one has to produce and the other has to
consume the species). Without this condition, the dynamics would simply cross the tropical
manifold with no deviation. However, the condition ii) is not sufficient for
stability of the equilibration (permanence). A sufficient stability condition
is given by Lemma~\ref{permcrit} and
reads $(\vect{\alpha_j} - \vect{\alpha'_j})_i (\vect{\beta_{j}} - \vect{\alpha_{j}})_i > 0.$

\section{Tropical approach to the permanency problem}

We have shown in the previous sections that tropical ideas can be used to simplify complex systems, by
tropical removing. During this procedure, permanency has to be checked at intermediate steps on
tropically truncated systems. Lemma~\ref{permcrit} allows to check permanency for toric systems
with separated time scales. We provide here another approach to permanency, that can be applied to more general
situations. We consider only upper estimates. The lower estimates can be found in a similar way.



Like in the preceding sections the truncated system is obtained by removing from
the non-tropicalized system \eqref{mainsystem} all the terms excepting the maximum
order terms. We denote the corresponding vector field by $F^{tr}$ and the
truncated differential equations read:
\begin{equation}
\frac{d x_i}{dt} = F_i^{tr}(\vect{x}). \label{chemtR}
\end{equation}
 Let us first formulate a  Lemma.
 \vspace{0.2cm}
 \begin{lemma}  Assume that non-tropicalized system \eqref{mainsystem} has a smooth Lyapunov function $V(\vect{x})$ defined on the cone
 ${\bf R}^n_{>}$ such that
 \begin{equation}
   dV(\vect{x}(t))/dt  \le 0
 \label{DV}
\end{equation}
 on trajectories $\vect{x}(t)=(x_1(t),...x_n(t))$ of \eqref{mainsystem} and
 \begin{equation}
   V(\vect{x}) \to \infty \ as \quad |\vect{x}| \to \infty.
 \label{VX}
\end{equation}
 Then, if $\vect{x}(t)$ is a trajectory of \eqref{mainsystem} such that $|\vect{x}(0)|<\delta'$, then there is a constant $C_0(\delta')$ such that
 \begin{equation}
   |\vect{x}(t)| < C_0, \quad t > 0.
 \label{VXb}
\end{equation}
\end{lemma}
\begin{proof}
 Indeed, if $|\vect{x}(t)|$ are unbounded as $t \to +\infty$, one has $\sup_{t>0} V(\vect{x}(t)) =+ \infty$, but
 (\ref{DV}) entails $V(\vect{x}(t)) \le V(\vect{x}(0))$.
\end{proof}
 Let us consider the tropical version \eqref{chemtR}. Assume that the truncated
  version \label{chemtR}
  has a
 {\em strong} Lyapunov function $V^{tr}(x)$.
 For a truncated vector field $F^{tr}$ this function satisfies
 \begin{equation}
   <\vect{\nabla} V^{tr}(\vect{x}(t)), \vect{F}^{tr}(\vect{x}(t))>  \le - \kappa |\vect{\nabla} V^{tr}(\vect{x}(t))||\vect{F}^{tr}(\vect{x}(t))|, \quad \kappa > 0
 \label{DVt}
\end{equation}
 on trajectories $\vect{x}(t)=(x_1(t),...x_n(t))$ of \eqref{mainsystem} and
 \begin{equation}
   V^{tr}(\vect{x}) \to \infty \quad as \quad |\vect{x}| \to \infty.
 \label{VXt}
\end{equation}
 Here $\vect{x} \in {\bf R}^n_{> }$.

 Such a function can be found for some tropical versions of two component systems. For example, if
 \begin{equation}
   dx/dt= k_1 x^a y^b,
 \label{DVt1}
\end{equation}
 \begin{equation}
   dy/dt= -k_2 x^a y^b,
 \label{DVt3}
 \end{equation}
 where $a, b > 0$ and $k_1, k_2 >0$, we can define $V^{tr}$ by
 \begin{equation}
   V^{tr}= x + \beta y,
 \label{DVt4}
 \end{equation}
 where $\beta k_2 > k_1$.  Then $\nabla V=(1, \beta)$, and (\ref{DVt}), (\ref{VXt}) hold.
 \begin{lemma}{\em Assume the tropicalized system (\ref{chemtR}) has a smooth Lyapunov function $V^{tr}(\vect{x})$ defined on the cone ${\bf R}^n_{>}$ such that  (\ref{DVt}), (\ref{VXt}) hold. Assume that $F_i(\vect{x})$ and $F_i^{tr}(\vect{x})$
 are multivariate polynomials of $\vect{x}$ such that $deg(F_i^{tr}) > deg(\tilde F_i)$, where
 $\tilde F_i = F_i - F_i^{tr}$.

 Then, if $\vect{x}(0) < \delta'$, then there is a constant $C_0(\delta')$ such that solutions of non-tropicalized system (\ref{chem}) satisfy
 \begin{equation}
   |\vect{x}(t)| < C_0, \quad t > 0.
 \label{VXtt}
\end{equation}
}
\end{lemma}
 \begin{proof}
 Let us compute the derivative $dV^{tr}/dt$ on trajectories of the initial (non-tropicalized) system. We have the relation
 \begin{equation}
   dV^{tr}(\vect{x}(t))/dt= < \vect{\nabla } V^{tr}(\vect{x}(t)) , \vect{F}^{tr} (\vect{x}(t)) >  +  < \vect{\nabla } V^{tr}(\vect{x}(t)) ,  \tilde{\vect{F}}(\vect{x}(t) ) > .
 \label{est1}
\end{equation}

 Using the definition of strong Lyapunov functions,  from (\ref{est1}) one has
 \begin{equation}
   dV^{tr}(\vect{x}(t))/dt \le |\vect{\nabla} V^{tr}(\vect{x}(t))|(-\kappa | \vect{F}^{tr}(\vect{x}(t))| +  |\tilde{\vect{F}}(\vect{x}(t))|).
 \label{est2}
\end{equation}
 But for large $|\vect{x}|$ one has $|\tilde{\vect{F}}(\vect{x})| < \kappa | \vect{F}^{tr}(\vect{x}(t))|$, because
 $|\tilde{\vect{ F}}(x)| =o(|\vect{F}^{tr}(\vect{x})|),\, |\vect{x}| \to \infty$. Therefore,
 (\ref{est2}) gives then
 \begin{equation}
   dV^{tr}(\vect{x}(t))/dt \le 0.
 \label{est3}
\end{equation}
 This shows that $|\vect{x}(t)|$ cannot increase to $+\infty$, and finishes the proof.
\end{proof}

\section{Application to chemical reactions kinetics}

As an application, we discuss the Michaelis-Menten mechanism of catalysed reaction. This model
can be schematically
described as~:
\begin{equation}
S + E \underset{k_{-1}}{ \overset{k_{1}}{\rightleftharpoons}} ES  \overset{k_2}{\rightarrow}
P + E,
\notag
\end{equation}
where $S,E,ES,P$ represent the substrate, the enzyme, the enzyme-substrate complex and the product,
respectively.

The rate functions obey mass-action laws. We denote by $x=[S]$ and $y=[SE]$, the concentration of substrate and enzyme-substrate complex, respectively. The reaction mechanism has two conserved quantities
$e_0 = [E] + [ES]$, $s_0 = [S] + [ES] + [P]$. Using the conservation laws
we obtain the following reduced system:
\begin{eqnarray}
 x' & = -k_1 x(e_0-y) + k_{-1} y, \notag \\
 y' & = k_1 x (e_0 - y) - (k_{-1}+k_2)y.
\label{mm}
\end{eqnarray}
Let us consider that the initial data satisfies $0 \leq y(0) \leq e_0$, $0 \leq x(0) + y(0) \leq s_0$. Then,
from \eqref{mm} it follows that
\begin{equation}
0 \leq y \leq e_0, \, 0 \leq x+y \leq s_0,\, 0\leq x.
\label{c0}
\end{equation}

This type of constraints are typical for reduced systems resulting from ODE systems with conservation.

We solve now the tropical equilibration problem. Using rescaled variables $x= \bar x \epsilon^{a_1}$, $y= \bar y \epsilon^{a_2}$,
$k_1= \bar k_1 \epsilon^{\gamma_1}$, $k_{-1}= \bar k_{-1} \epsilon^{\gamma_{-1}}$, $e_0= \bar e_0 \epsilon^{\gamma_e}$, $s_0= \bar s_0 \epsilon^{\gamma_s}$,
\eqref{mm} becomes:
\begin{eqnarray}
 \bar x' & = - \bar k_1 \bar e_0 \epsilon^{\gamma_{1}+ \gamma_{e}} \bar x +
 \bar k_1 \epsilon^{\gamma_{1}+ a_2} \bar x \bar y  +
 \bar k_{-1} \epsilon^{\gamma_{-1}+ a_2 - a_1} \bar y, \notag \\
 \bar y' & = \bar k_1 \bar e_0 \epsilon^{\gamma_{1}+ \gamma_{e} + a_1 - a_2} \bar x -
 \bar k_1 \epsilon^{\gamma_{1}+ a_1} \bar x \bar y  -
 (\bar k_{-1} \epsilon^{\gamma_{-1}} + \bar k_2 \epsilon^{\gamma_{2}})\bar y.
\label{mmr}
\end{eqnarray}
The two tropical equilibration equations for $x$ and $y$ read:
\begin{eqnarray}
\gamma_{1}+ \gamma_{e} &= \min (\gamma_{1}+ a_{2},\gamma_{-1}+ a_2 - a_1), \label{eq1} \\
\gamma_{1}+ \gamma_{e} + a_1 - a_2 &= \min (\gamma_{1}+ a_1,\min(\gamma_{-1},\gamma_2)). \label{eq2}
\end{eqnarray}

We should add to these, the constraints \eqref{c0} imposed by the dynamics:
\begin{equation}
a_2 \geq \gamma_e, \, \min(a_1,a_2) \geq \gamma_s.
\end{equation}
We can distinguish between two situations.

Let us first consider that  $\gamma_{-1} < \gamma_2$. In this case \eqref{eq2} is equivalent to
\eqref{eq1} (it can be derived from the latter by adding $a_1 - a_2$ to both sides).
This situation corresponds to $k_{-1}$ much larger than $k_2$ and means that the enzyme-substrate
complex is recycled to a much larger extent than it is transformed into the reaction product.
We can find
two solutions for the tropical equilibration problem and two different tropically truncated systems (TTS).

The first solution demands large concentrations of substrate and corresponds to
saturation of the enzyme (saturation regime):
\begin{eqnarray}
&a_1  <   \gamma_{-1} - \gamma_{1}, \, a_2 = \gamma_e, \notag  \\
& \bar x'  = \epsilon^{\gamma_{1}+ \gamma_{e}} (-\bar k_1 \bar e_0 \bar x  + \bar k_{1} \bar x \bar y), \notag \\
& \bar y'  = \epsilon^{\gamma_{1}+ a_1}  (\bar k_1 \bar e_0 \bar x  -  \bar k_{1} \bar x \bar y).
\label{mm1}
\end{eqnarray}
The second solution works for small concentrations of substrate (linear regime):
\begin{eqnarray}
& a_1  >   \gamma_{-1} - \gamma_{1}, \, a_2 = a_1 + \gamma_e + \gamma_{1} - \gamma_{-1}, \notag \\
& \bar x'  = \epsilon^{\gamma_{1}+ \gamma_{e}} (-\bar k_1 \bar e_0 \bar x  + \bar k_{-1} \bar y), \notag \\
& \bar y'  = \epsilon^{\gamma_{-1}}  (\bar k_1 \bar e_0 \bar x  -  \bar k_{-1} \bar y).
\label{mm2}
\end{eqnarray}
In order to further characterize these two functioning regimes, we consider the third variable $\bar z=(x+y)\epsilon^{-\gamma_s}$. The choice of this variable is dictated by the TTS. In general,
conserved quantities of the TTS (total mass of fast cycles) can be slow variables of
the full system \cite{radulescu2012red}.
If this variable is slower than both $x$ and $y$, the regime is called quasi-equilibrium
\cite{gorban2009asymptotology,gorban2011MM}
and consists in rapid exchanges
between substrate and enzyme and a much slower transformation of the total mass $[S] + [SE]$ into $[P]$. In
both cases the equation for $\bar z$ reads:
\begin{equation}
 \bar z' = - \epsilon^{\gamma_{2}+a_2 - \gamma_s}  \bar k_2 \bar y.
\end{equation}
A sufficient condition for quasi-equilibrium (ensuring both $\gamma_{2}+a_2 - \gamma_s > max(\gamma_{1}+ \gamma_{e},\gamma_{1}+ a_1)$ and $\gamma_{2}+a_2 - \gamma_s > max(\gamma_{1}+ \gamma_{e},\gamma_{-1})$
in the first and second of the cases above, respectively)
 is $\gamma_2 > \gamma_1+\gamma_s$.

The second situation is when  $\gamma_{-1} > \gamma_2$.
This case leads to negligible recycling of the enzyme-substrate complex that
is rapidly transformed into reaction product. Quasi-equilibrium is no longer possible,
but we have another interesting equilibration corresponding to fast consumption
of one of the variables. The QSS variable is necessarily equilibrated and fast. The
remaining variable is slow.
This corresponds to the well known quasi-steady state (QSS) regime of
the Michaelis-Menten mechanism, first discussed by Briggs and Haldane \cite{gorban2009asymptotology,gorban2011MM,radulescu2012red}.

In this case \eqref{eq1},\eqref{eq2} are no
longer equivalent:
\begin{eqnarray}
\gamma_{1}+ \gamma_{e} &= \min (\gamma_{1}+ a_{2},\gamma_{-1}+ a_2 - a_1), \label{eq1bis} \\
\gamma_{1}+ \gamma_{e} + a_1 - a_2 &= \min (\gamma_{1}+ a_1,\gamma_2). \label{eq2bis}
\end{eqnarray}
We obtain four solutions to the tropical equilibration problem and four different truncated systems.
In three of these solutions, only one variable is equilibrated (see Table 1). The solutions 1 and 4 correspond to rapid complex consumption in saturated and linear regimes, respectively.
It is the case discussed by Briggs and Haldane. The solutions 2 and 3 correspond to very small concentrations of the substrate.
\begin{table}[h]
\centering
\caption{\small Tropical equilibrations of the Michaelis-Menten model with negligible recycling $\gamma_2 < \gamma_{-1}$. All these equilibrations have a geometrical interpretation illustrated in Figure 1.}
\begin{tabular}{l|l|l|l}
\hline
No & Condition & Truncated system & Regime \\
\hline
& &  &\\
1 & $a_1 <   \gamma_{2} - \gamma_{1}$
 & $x' = \epsilon^{\gamma_{1}+ \gamma_{e}} (-{\bar k}_1 \bar e_0 \bar x  + {\bar k}_{1} \bar x \bar y)$ & $y$ QSS if \\
& $a_2 = \gamma_e$ &
$y'  = \epsilon^{\gamma_{1}+a_1}  ({\bar k}_1 \bar e_0 \bar x  -  {\bar k}_{1} \bar x \bar y)$ & $a_1 < \gamma_e$  \\
\hline
& & & \\
2 & $\gamma_{2} - \gamma_{1} <  a_1   <   \gamma_{-1} - \gamma_{1}$ &
$x'  = \epsilon^{\gamma_{1}+ \gamma_{e}} (-{\bar k}_1 \bar e_0 \bar x  + {\bar k}_{1} \bar x \bar y)$
& $x$ QSS if \\
&  $a_2 = \gamma_e$ & $y'  = - \epsilon^{\gamma_{2}} {\bar k}_{2} \bar y$  &
$\gamma_2 > \gamma_1+\gamma_e$
 \\
\hline
& & &\\
3 & $a_1   >   \gamma_{-1} - \gamma_{1}$ &
 $x'  = \epsilon^{\gamma_{1}+ \gamma_{e}} (-\bar k_1 \bar e_0 \bar x  + \bar k_{-1} \bar y)$
& $x$ QSS if   \\
& $a_2 = a_1 + \gamma_e + \gamma_{1} - \gamma_{-1}$ & $y'  = - \epsilon^{\gamma_{2}} \bar k_{2} \bar y$ &
$\gamma_2 > \gamma_1+\gamma_e$
\\
\hline
& & &\\
4 & $a_1   >   \gamma_{2} - \gamma_{1}$
& $x'  = - \epsilon^{\gamma_{1}+ \gamma_{e}} \bar k_1 \bar e_0 \bar x$ &
$y$ QSS if \\
& $a_2 = a_1 + \gamma_e + \gamma_{1} - \gamma_{2}$ &  $y'  = \epsilon^{\gamma_{2}}  (\bar k_1 \bar e_0 \bar x  -  \bar k_{2} \bar y)$ &
$\gamma_2 < \gamma_1+\gamma_e$
\\
\hline
\end{tabular}
\end{table}
\begin{figure}[h!]
\centerline{
\includegraphics[width=150mm]{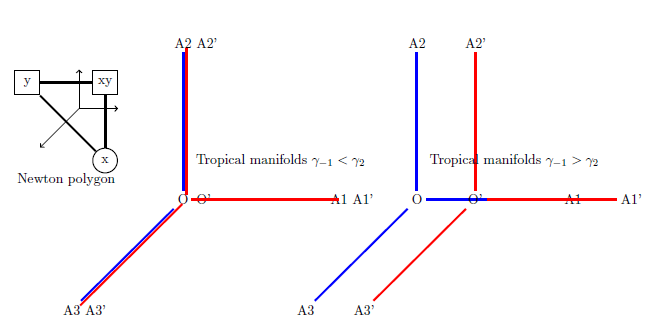}}
\caption{\small
Newton polygon and tropical manifolds for $x$ (in blue) and for $y$ (in red)
variables of the Michaelis-Menten model.
If $\gamma_{-1} < \gamma_2$ the two manifolds
coincide in the limit $\epsilon \to 0$, meaning that both variables are equilibrated.
The vertices of the Newton polygons correspond to monomial terms in the ODEs
(different vertex shapes mean different signs of the monomials).
Only two edges of the Newton polygon relates vertices of opposite signs,
which means that there are two equilibrations; these correspond
to the branches $OA_3 \equiv O'A_3$ and $OA_1 \equiv O'A'_1$ of the tropical manifolds.
If $\gamma_{-1} > \gamma_2$, the two tropical manifolds for $x$ and for $y$ share a common
half-line, but no longer coincide. This leads to four possible equilibrations
as in Table 1: $O'A'1$ (solution 1), $OO'$ (solution 2), $OA_3$ (solution 3),
and $O'A'_3$ (solution 4).}
\label{fig_MM_tropical}
\end{figure}

\section{Conclusion}
Tropical analysis provides useful tools for understanding the dynamics
of biochemical networks.
In this paper we have studied the simple example of an enzymatic reaction, but
some other applications have been discussed elsewhere, see \cite{SASB2011,savageau2009phenotypes}.
We have shown that depending on the values of the parameters and concentrations,
biochemical networks with multiple time scales can have several asymptotic regimes.
During such regimes, the dynamics can be approximated by truncated systems obtained
by tropicalization of the ordinary differential equations describing the chemical
kinetics. Tropical geometry can guide the construction of such truncated systems.
An important step in this construction is the calculation of tropical equilibrations
leading to slow invariant manifolds. We showed that there is one to one correspondence
between tropical equilibrations and well defined parts of the tropical manifolds
of the polynomials defining the ordinary differential equations. In the future,
effective algorithms will be needed for the tropical equilibration problem, whose
complexity is exponential in the number of variables. This will be essential for
large scale applications in systems biology. Also, our methods will be generalized to
include the case  when the tropically truncated fast subsystem has conservation laws,
when aggregated slow variables are needed.

\section*{Acknowledgements}
D.G. is grateful to the Max-Planck Institut f\"ur Mathematik, Bonn for its hospitality
and to  Labex CEMPI (ANR-11-LABX-0007-01). O.R. gratefully acknowledges
financial support from a CNRS/INRIA/INSERM grant (PEPS BMI) and from ANR
(project Biotempo). We are thankful for useful comments from an anonymous Referee.

\newcommand{\etalchar}[1]{$^{#1}$}
\providecommand{\bysame}{\leavevmode\hbox to3em{\hrulefill}\thinspace}
\providecommand{\MR}{\relax\ifhmode\unskip\space\fi MR }
\providecommand{\MRhref}[2]{%
  \href{http://www.ams.org/mathscinet-getitem?mr=#1}{#2}
}
\providecommand{\href}[2]{#2}

\end{document}